\documentclass[journal,twoside,web]{ieeecolor}
\usepackage{lcsys}
\usepackage{cite}
\usepackage{amsmath,amssymb,amsfonts}
\usepackage{algorithmic}
\usepackage{graphicx}
\usepackage{textcomp}

\usepackage{mathrsfs}
\usepackage{mathtools}
\usepackage{bm}

\newcommand{\defref}[1]{Definition~\ref{#1}}
\newcommand{\asmref}[1]{Assumption~\ref{#1}}

\newcommand{\lemref}[1]{Lemma~\ref{#1}}
\newcommand{\thmref}[1]{Theorem~\ref{#1}}

\newcommand{\remref}[1]{Remark~\ref{#1}}

\newcommand{\secref}[1]{Section~\ref{#1}}
\newcommand{\figref}[1]{Fig.~\ref{#1}}

\newcommand{\N}{\mathbb{N}}
\newcommand{\Z}{\mathbb{Z}}
\newcommand{\R}{\mathbb{R}}
\newcommand{\X}{\mathcal{X}}
\newcommand{\K}{\mathcal{K}}
\newcommand{\M}{\mathcal{M}}
\newcommand{\C}{\mathcal{C}}
\newcommand{\KeyGen}{\mathsf{KeyGen}}
\newcommand{\Enc}{\mathsf{Enc}}
\newcommand{\Dec}{\mathsf{Dec}}
\newcommand{\Eval}{\mathsf{Eval}}
\newcommand{\KeyUpd}{\mathsf{KeyUpd}}
\newcommand{\CtUpd}{\mathsf{CtUpd}}
\newcommand{\EC}{\mathsf{EC}}
\newcommand{\Ecd}{\mathsf{Ecd}}
\newcommand{\Dcd}{\mathsf{Dcd}}
\newcommand{\Sum}{\mathsf{Sum}}
\newcommand{\pk}{\mathsf{pk}}
\newcommand{\sk}{\mathsf{sk}}
\newcommand{\ct}{\mathsf{ct}}
\newcommand{\tr}{\mathop{\mathrm{tr}}\limits}
\newcommand{\diag}{\mathop{\mathrm{diag}}\limits}
\newcommand{\argmin}{\mathop{\mathrm{arg~min}}\limits}
\newcommand{\EV}{\mathop{\mathbb{E}}\limits}

\newtheorem{definition}{Definition}
\newtheorem{assumption}{Assumption}

\newtheorem{lemma}{Lemma}
\newtheorem{theorem}{Theorem}

\newtheorem{remark}{Remark}

\allowdisplaybreaks[4]

\def\BibTeX{{\rm B\kern-.05em{\sc i\kern-.025em b}\kern-.08em
    T\kern-.1667em\lower.7ex\hbox{E}\kern-.125emX}}
\begin{document}
\title{Optimal Controller and Security Parameter for Encrypted Control Systems Under Least Squares Identification}
\author{Kaoru Teranishi, \IEEEmembership{Graduate Student Member, IEEE} and Kiminao Kogiso \IEEEmembership{Member, IEEE}
\thanks{This work was supported by JSPS Grant-in-Aid for JSPS Fellows Grant Number JP21J22442 and JSPS KAKENHI Grant Number JP22H01509.}
\thanks{K. Teranishi and K. Kogiso are with the Department of Mechanical and Intelligent Systems Engineering, The University of Electro-Communications, 1-5-1 Chofugaoka, Chofu, Tokyo 1828585, Japan (e-mail: teranishi@uec.ac.jp, kogiso@uec.ac.jp).}
\thanks{K. Teranishi is also with Japan Society for the Promotion of Science, Chiyoda-ku, Tokyo 1020083, Japan.}
}

\thispagestyle{empty}
\hspace{-4.5mm}
\fbox{
\begin{minipage}{\textwidth-5mm}\scriptsize
© 20XX IEEE.
Personal use of this material is permitted.
Permission from IEEE must be obtained for all other uses, in any current or future media, including reprinting/republishing this material for advertising or promotional purposes, creating new collective works, for resale or redistribution to servers or lists, or reuse of any copyrighted component of this work in other works.
\end{minipage}
}
\newpage
\setcounter{page}{0}

\maketitle
\thispagestyle{empty}

\begin{abstract}
Encrypted control is a framework for the secure outsourcing of controller computation using homomorphic encryption that allows to perform arithmetic operations on encrypted data without decryption.
In a previous study, the security level of encrypted control systems was quantified based on the difficulty and computation time of system identification.
This study investigates an optimal design of encrypted control systems when facing an attack attempting to estimate a system parameter by the least squares method from the perspective of the security level.
This study proposes an optimal $\mathnormal{H_2}$ controller that maximizes the difficulty of estimation and an equation to determine the minimum security parameter that guarantee the security of an encrypted control system as a solution to the design problem.
The proposed controller and security parameter are beneficial for reducing the computation costs of an encrypted control system, while achieving the desired security level.
Furthermore, the proposed design method enables the systematic design of encrypted control systems.
\end{abstract}


\section{Introduction}
\label{Introduction}

\IEEEPARstart{E}{ncrypted} control using homomorphic encryption is a major approach for security enhancement of networked control systems, as a network eavesdropper and a controller server cannot learn about the control system~\cite{darup2021,kogiso2015}.
Unlike traditional public-key encryption, homomorphic encryption enables arithmetic operations on encrypted data, and therefore the server does not require a secret key for decryption.
Hence, encrypted control has been applied to various controls, as in~\cite{farokhi2017,kim2023,darup2018a,alexandru2020b,suh2021}, to realize the secure outsourcing of controller computation to an untrusted server and implemented to some practical systems~\cite{cheon2018a,teranishi2020,shono2022}.
Moreover, attacks for encrypted control systems and their countermeasures were studied in~\cite{fauser2022,naseri2022,alisic2023}.

Although most existing encrypted controls rely on the security of used homomorphic encryption, the controls need other security definitions because, in control systems, the information to be protected is system parameters rather than a single sensor or control signal data at a certain point of time.
To solve this problem, recent studies have explored the security of encrypted control systems.
In~\cite{teranishi2022a}, the authors examined the provable security of the systems and analyzed the connection between the security and a traditional cryptographic security definition.
The study~\cite{teranishi2023a} focused on quantifying the security level of encrypted control systems using the sample complexity and computation time of system identification, disclosing the parameters of a target system.
The study also included a design for a controller that maximizes the sample complexity of the Bayes estimation for a system matrix of a closed-loop system with an encrypted controller.
Then, the study determined the minimum key length required to ensure that the computation time exceeds the period in which the target system is replaced.
Additionally, for a given security parameter, the study~\cite{kim2020a} provided a guideline for choosing cryptosystem parameters in an encrypted control system.

Here we consider the design of encrypted control systems under the least squares identification attack for a system matrix of a closed-loop system.
Preventing such attacks is essential to realizing secure control systems because once the attack is successful, an attacker can implement undetectable attacks based on the system model~\cite{Chong2019}.
A major challenge in designing the systems against the attack is the computation costs associated with encryption algorithms~\cite{darup2021}.
The use of homomorphic encryption can significantly increase the computational burden on the system, leading to longer computation times, potentially affecting the real-time performance of the system.
Furthermore, as the security parameter increases, the computation costs also increase, which leads to a trade-off between security level and performance.

This study proposes a systematic method for solving the security and performance trade-off by designing an optimal controller and security parameter for encrypted control systems based on the security definition in~\cite{teranishi2023a}.
To this end, this study derives a novel sample complexity of the systems.
With the novel sample complexity, we reveal that the security level of the system is connected to the controllability Gramian of the target system.
The optimal controller is designed as an optimal $H_2$ controller that minimizes the trace of the controllability Gramian to maximize the security level for a given security parameter.
Then, the optimal security parameter is determined as the minimum security parameter to achieve the desired security level.

The proposed method contributes to the generalization of the design method in~\cite{teranishi2023a}.
The previous method chose the minimum key length for a specific encryption scheme.
In contrast, the proposed method determines the minimum security parameter rather than the key length.
A security parameter is a common quantity for encryption schemes, and thus the proposed method can be applied to encrypted control systems with any homomorphic encryption.
Moreover, the attack based on the least squares method considered in this study is easier for attackers to perform compared to the Bayesian estimation in~\cite{teranishi2023a} because the method does not require prior knowledge of a target system.
Therefore, a broader class of encrypted control systems can be protected by preventing least squares identification attacks.

The rest of this paper is organized as follows.
\secref{sec:preliminaries} defines the syntax of homomorphic encryption and encrypted control.
\secref{sec:security} formulates an attack scenario and introduces the security definition of encrypted control systems.
\secref{sec:optimal_design} proposes an optimal controller and security parameter.
\secref{sec:example} shows a numerical example.
\secref{sec:conclusions} describes conclusions and future work.

\section{Preliminaries}
\label{sec:preliminaries}

\subsection{Notation}

The sets of natural numbers, integers, and real numbers are denoted by $\N$, $\Z$, and $\R$, respectively.
A key space, a plaintext space, and a ciphertext space are denoted by $\K$, $\M$, and $\C$, respectively.
Define the set $\Z^+ \coloneqq \{z \in \Z \mid 0 \le z\}$ and the bounded set $\X \subset \R$.
The sets of $n$-dimensional vectors and $m$-by-$n$ matrices of which elements and entries belonging to the set $\mathcal{A}$ are denoted by $\mathcal{A}^n$ and $\mathcal{A}^{m \times n}$, respectively.
The $i$th element of vector $v \in \mathcal{A}^n$ and the $(i,j)$ entry of matrix $M \in \mathcal{A}^{m \times n}$ are denoted by $v_i$ and $M_{ij}$, respectively.
The Frobenius norm of $M \in \mathcal{A}^{m \times n}$ is denoted by $\|M\|_F \coloneqq \sqrt{\tr(M^\top M)}$.

\subsection{Homomorphic encryption}

This section describes the syntax and security level of encryption.
In the following, a security parameter is denoted by $\lambda \in \N$.
First, homomorphic encryption is defined as follows~\cite{Acar18}.

\begin{definition}
\label{def:he}
    Homomorphic encryption is $(\KeyGen, \allowbreak \Enc, \allowbreak \Dec, \allowbreak \Eval)$ such that:
    \begin{itemize}
        \item $(\pk, \sk) \gets \KeyGen(1^\lambda)$: A key generation algorithm takes $1^\lambda$ and outputs a key pair $(\pk, \sk) \in \K$, where $1^\lambda$ is the unary representation of a security parameter, $\pk$ is a public key, and $\sk$ is a secret key.
        \item $\ct \gets \Enc(\pk, m)$: An encryption algorithm takes a public key $\pk$ and a plaintext $m \in \M$ and outputs a ciphertext $\ct \in \C$.
        \item $m \gets \Dec(\sk, \ct)$: A decryption algorithm takes a secret key $\sk$ and a ciphertext $\ct \in \C$ and outputs a plaintext $m \in \M$.
        \item $\ct \gets \Eval(\pk, \ct_1, \ct_2)$: A homomorphic evaluation algorithm takes a public key $\pk$ and ciphertexts $\ct_1, \ct_2 \in \C$ and outputs a ciphertext $\ct \in \C$.
        \item Correctness: $\Dec(\sk, \Enc(\pk, m)) = m$ holds for any $(\pk, \sk) \gets \KeyGen(1^\lambda)$ and for any $m \in \M$.
        \item Homomorphism: $\Dec(\sk, \Eval(\pk, \ct_1, \ct_2)) = m_1 \bullet m_2$ holds for any $(\pk, \sk) \gets \KeyGen(1^\lambda)$ and for any $m_1, m_2 \in \M$, where $\ct_1 \gets \Enc(\pk, m_1)$, $\ct_2 \gets \Enc(\pk, m_2)$, and $\bullet$ is a binary operation on $\M$.
    \end{itemize}
\end{definition}

Homomorphic encryption is called as additive, multiplicative, or (leveled) fully homomorphic encryption if the binary operation is addition ($\bullet = +$), multiplication ($\bullet = \times$), or both addition and multiplication, respectively.

Next, we define updatable homomorphic encryption.

\begin{definition}
\label{def:upd_he}
    Let $\Pi=(\KeyGen, \allowbreak \Enc, \allowbreak \Dec, \allowbreak \Eval)$ be homomorphic encryption.
    Updatable homomorphic encryption is $(\Pi, \allowbreak \KeyUpd, \allowbreak \CtUpd)$ such that:
    \begin{itemize}
        \item $(\pk_{t+1}, \sk_{t+1}, \sigma_t) \gets \KeyUpd(\pk_t, \sk_t)$: A key update algorithm takes a key pair $(\pk_t, \sk_t) \in \K$ at time $t \in \Z^+$ and outputs an updated key pair $(\pk_{t+1}, \sk_{t+1}) \in \K$ and an update token $\sigma_t$.
        \item $\ct_{t+1} \gets \CtUpd(\ct_t, \sigma_t)$: A ciphertext update algorithm takes a ciphertext $\ct_t \in \C$ and an update token $\sigma_t$ at time $t \in \Z^+$ and outputs an updated ciphertext $\ct_{t+1} \in \C$.
        \item Correctness: $\Dec(\sk_t, \ct_t) = \Dec(\sk_t, \Enc(\pk_t, m)) = m$ holds for any $(\pk_0, \sk_0) \gets \KeyGen(1^\lambda)$, for any $m \in \M$, and for all $t \in \Z^+$, where $\ct_0 \gets \Enc(\pk_0, m)$, $(\pk_{t+1}, \sk_{t+1}, \sigma_t) \gets \KeyUpd(\pk_t, \sk_t)$, and $\ct_{t+1} \gets \CtUpd(\ct_t, \sigma_t)$.
        \item Homomorphism: $\Dec(\sk_t, \Eval(\pk_t, \ct_{1,t}, \ct_{2,t})) = \Dec(\sk_t, \Eval(\pk_t, \Enc(\pk_t, m_1), \Enc(\pk_t, m_2))) = m_1 \bullet m_2$ holds for any $(\pk_0, \sk_0) \gets \KeyGen(1^\lambda)$, for any $m_i \in \M$, and for all $t \in \Z^+$, where $\ct_{i,0} \gets \Enc(\pk_0, m_i)$, $(\pk_{t+1}, \sk_{t+1}, \sigma_t) \gets \KeyUpd(\pk_t, \sk_t)$, $\ct_{i,t+1} \gets \CtUpd(\ct_{i,t}, \sigma_t)$, and $i = 1, 2$.
    \end{itemize}
\end{definition}

Updatable homomorphic encryption is a public-key variant of updatable encryption~\cite{boneh2013,lehmann2018} with a homomorphic evaluation algorithm.
The following property is assumed for the updatable homomorphic encryption used in this study.

\begin{assumption}
\label{asm:key_upd}
    A key pair $(\pk_k, \sk_k)$ provides no information about a key pair $(\pk_j, \sk_j)$ for any $k \in \Z^+$ and for any $j \in \Z^+ \setminus \{k\}$, where $(\pk_0, \sk_0) \gets \KeyGen(1^\lambda)$, and $(\pk_{t+1}, \sk_{t+1}, \sigma_t) \gets \KeyUpd(\pk_t, \sk_t)$.
\end{assumption}

Although one may think that the assumption is significantly stronger than a single-key case, updatable homomorphic encryption scheme satisfying it can be realized based on a standard cryptographic assumption~\cite[Propositions~2 and 3]{teranishi2023a}.

Finally, we define the security level of encryption schemes, which is quantified using the number of bits~\cite{Katz21}.

\begin{definition}
\label{def:bit_sec}
    An encryption scheme satisfies $\lambda$ bit security if at least $2^\lambda$ operations are required to break the scheme.
\end{definition}

\begin{remark}
    By \defref{def:bit_sec}, a security parameter $\lambda$ represents the security level of an encryption scheme.
    The optimal key length $k^\ast$ of an encryption scheme satisfying $\lambda$ bit security can be computed as
    \begin{equation}
        k^\ast = \argmin{}_{k\in\N}\, \Omega(k) \quad \text{s.t.} \quad \Omega(k) \ge 2^{\lambda},
        \label{eq:opt_key}
    \end{equation}
    where $\Omega(k)$ is the time complexity of the fastest known algorithm for breaking the encryption scheme.
\end{remark}

\subsection{Encrypted control}

Using (updatable) homomorphic encryption, encrypted control is defined as follows.

\begin{definition}
\label{def:ec}
    Given (updatable) homomorphic encryption and a controller $f:(\Phi, \xi) \mapsto \psi$, where $\Phi \in \X^{\alpha \times \beta}$ is a controller parameter, $\xi \in \X^\beta$ is a controller input, and $\psi \in \X^\alpha$ is a controller output.
    Let there exist $\Ecd$ and $\Dcd$ such that:
    \begin{itemize}
        \item $m \gets \Ecd(x; \Delta)$: An encoder algorithm takes $x \in \X$ and a scaling factor $\Delta \in \R$ and outputs a plaintext $m \in \M$.
        \item $x \gets \Dcd(m; \Delta)$: A decoder algorithm takes a plaintext $m \in \M$ and a scaling factor $\Delta \in \R$ and outputs $x \in \X$.
    \end{itemize}
    An encrypted controller of $f$ is $\EC$ such that:
    \begin{itemize}
        \item $\ct_\psi \gets \EC(\pk, \ct_\Phi, \ct_\xi)$: An encrypted control algorithm takes a public key $\pk$ and ciphertexts $\ct_\Phi \in \C^{\alpha \times \beta}, \ct_\xi \in \C^\beta$ and outputs a ciphertext $\ct_\psi \in \C^\alpha$.
        \item $\Dcd(\Dec(\sk, \EC(\pk, \ct_\Phi, \ct_\xi)); \Delta) \simeq f(\Phi, \xi)$ holds for some $\Delta \in \R$, for any $(\pk, \sk) \gets \KeyGen(1^\lambda)$, for any $\Phi \in \X^{\alpha \times \beta}$, and for any $\xi \in \X^\beta$, where $\ct_\Phi \gets \Enc(\pk, \Ecd(\Phi; \Delta))$, $\ct_\xi \gets \Enc(\pk, \Ecd(\xi; \Delta))$, and the algorithms perform each element of matrices and vectors.
    \end{itemize}
\end{definition}

Note that the controller parameter of an encrypted controller with additive homomorphic encryption is a plaintext rather than a ciphertext~\cite{farokhi2017,darup2018a}.
In this case, \defref{def:ec} can be modified by replacing $\ct_\Phi$ with $\Ecd(\Phi;\Delta)$.

\begin{remark}
\label{rem:quantization}
    Although encoder and decoder algorithms generally induce a quantization error $e(\Delta)$, i.e., $\Dcd(\Ecd(x; \Delta); \Delta) = x + e(\Delta)$, this study assumes that the error is negligible because of the appropriately chosen scaling factor $\Delta$~\cite{darup2021,teranishi2023a}.
\end{remark}

\section{Attack Scenario and Security Definition}
\label{sec:security}

This section formulates an attack scenario considered in this study and defines the security of encrypted control systems under the scenario.

\subsection{Attack scenario}

Given the control system
\begin{subequations}\label{eq:cs}
    \begin{align}
        x_{t+1} &= A_p x_t + B_p u_t + w_t, \label{eq:plant} \\
        u_t     &= F x_t,                   \label{eq:ctrl}
    \end{align}
\end{subequations}
where $t \in \Z^+$ is a time, $x_t \in \R^n$ is a state, $u_t \in \R^m$ is an input, and $w_t \in \R^n$ is a noise.
Suppose $x_0$ and $w_t$ are independent and identically distributed over the Gaussian distribution with mean $\bm{0}$ and variance $\sigma^2 I$.
The plant parameters $(A_p, B_p)$ are controllable, and $F$ is a feedback gain designed such that $A_p + B_p F$ is stable.
If $F \in \X^{m \times n}$, $x_t \in \X^n$, and $u_t \in \X^m$ for all $t \in \Z^+$, the encrypted control system of \eqref{eq:cs} with updatable multiplicative homomorphic encryption is given as
\begin{equation}
    \begin{aligned}
        x_{t+1}   &=     A_p x_t + B_p u_t + w_t,              \\
        u_t       &\gets \Dcd(\Dec(\sk_t, \ct_{u,t}); \Delta), \\
        \ct_{u,t} &\gets \EC(\pk_t, \ct_{F,t}, \ct_{x,t}),     \\
        \ct_{x,t} &\gets \Enc(\pk_t, \Ecd(x_t; \Delta)), 
    \end{aligned}
    \label{eq:ecs}
\end{equation}
where $(\pk_0, \sk_0) \gets \KeyGen(1^\lambda)$, $\ct_{F,0} \gets \Enc(\pk_0, \allowbreak \Ecd(F; \allowbreak \Delta))$, $(\pk_{t+1}, \sk_{t+1}, \sigma_t) \gets \KeyUpd(\pk_t, \sk_t)$, $\ct_{F,t+1} \gets \CtUpd(\ct_{F,t}, \sigma_t)$, $\EC$ is the encrypted controller of \eqref{eq:ctrl} that outputs a ciphertext matrix $\ct_{u,t} \in \C^{m \times n}$ of which $(i, j)$ entry is an output of $\Eval(\pk_t, \ct_{F_{ij},t}, \ct_{x_j,t})$, and the decryption algorithm is redefined as $\Dec \coloneqq \Sum \circ \Dec$ using $\Sum: \M^{m \times n} \to \M^m : M \mapsto [\, \sum_{i=1}^n M_{1i} \ \allowbreak \cdots \ \allowbreak \sum_{i=1}^n M_{mi} \,]^\top$~\cite{kogiso2015}.
Then, by \defref{def:ec}, the dynamics of the closed-loop system is given as
\begin{align}
       x_{t+1}
    &= A_p x_t + B_p \Dcd(\Dec(\sk_t, \ct_{u,t}); \Delta) + w_t, \nonumber \\
    &= A x_t + w_t, \label{eq:system}
\end{align}
where $A = A_p + B_p F$.
Note that $\Dcd(\Dec(\sk_t, \ct_{u,t}); \Delta) = F x_t$ holds thanks to the assumption in \remref{rem:quantization}.

Given the above settings and updatable homomorphic encryption satisfying \asmref{asm:key_upd}, this study considers the following attack scenario.

\begin{definition}
\label{def:attack}
    The attacker follows the procedure below.
    \begin{enumerate}
        \item
        The attacker eavesdrops the ciphertexts $\ct_{x,t}$ of \eqref{eq:ecs} within $t \in [t_s, t_f]$, where $0 < t_s < t_f < \infty$.
        \item
        The attacker deciphers the ciphertexts to obtain the original data $\{x_{t_s},\cdots,x_{t_f}\}$.
        \item
        The attacker estimates $A$ of \eqref{eq:system} by the least squares method,
        \begin{equation}
            \hat{A} = \argmin{}_{A \in \R^{n \times n}} \| X_f - A X_p \|_F^2 = X_f X_p^+,
            \label{eq:estimate}
        \end{equation}
        where $X_f = A X_p + W_p$, $X_f = [x_{t_s+1} \ \cdots \ x_{t_f}]$, $X_p = [x_{t_s} \ \cdots \ x_{t_f-1}]$, $W_p = [w_{t_s} \ \cdots \ w_{t_f-1}]$, and $X_p^+$ is the pseudo inverse matrix of $X_p$.
        We assume that $X_p$ is full row rank throughout this paper.
        Note that the assumption is met for a sufficiently large sample size $N = t_f - t_s + 1$ in practice.
    \end{enumerate}
\end{definition}

Additionally, we define the estimation error as
\begin{equation}
    \epsilon(N, F) \coloneqq (1 / n^2) \| A - \hat{A} \|_F^2.
    \label{eq:error}
\end{equation}
It should be noted that the estimation error implicitly depends on $F$ of \eqref{eq:ctrl} because $A$ of \eqref{eq:system} can be tuned by designing $F$.
This fact is relevant later in the controller design.

\begin{remark}
\label{rem:ecs}
    With typical multiplicative or (leveled) fully homomorphic encryption~\cite{ElGamal85,ckks2017}, the encrypted control system of \eqref{eq:cs} is given as \eqref{eq:ecs}, replacing $(\pk_t, \sk_t)$ and $\ct_{F,t}$ with $(\pk, \sk) = (\pk_0, \sk_0)$ and $\ct_F = \ct_{F,0}$, respectively.
    In addition, $\ct_F$ is modified to $\Ecd(F; \Delta)$ when using typical additive homomorphic encryption~\cite{Paillier99,regev2005}.
    The attack in \defref{def:attack} can be applied even to such encrypted control systems because the closed-loop dynamics of the systems are represented by \eqref{eq:system}.
\end{remark}

\begin{remark}
    The required computation time to perform the second step in \defref{def:attack} is determined by a security parameter, which will be formulated in \defref{def:sdt} to define the security of encrypted control systems.
\end{remark}

\subsection{Security of encrypted control system}

This study employs the security definition in~\cite{teranishi2023a} for encrypted control systems under the attack in \defref{def:attack}.
Roughly speaking, in the definition, an encrypted control system is said to be secure if an attacker cannot estimate the parameters of a target system with a certain accuracy within a given period.
The security is formulated based on two quantities, \textit{sample identifying complexity} and \textit{sample deciphering time}, defined below.

\begin{definition}
\label{def:sic}
    Let $N$ be a sample size.
    A sample identifying complexity of \eqref{eq:system} under the attack in \defref{def:attack} is a function $\gamma$ satisfying $\gamma(N, F) \le \EV[\epsilon(N, F)]$, where $F$ and $\epsilon$ are defined in \eqref{eq:ctrl} and \eqref{eq:error}, respectively.
\end{definition}

\begin{definition}
\label{def:sdt}
    A sample deciphering time is a computation time $\tau$ required for breaking $N$ ciphertexts of an \textit{updatable} homomorphic encryption scheme that satisfies $\lambda$ bit security and \asmref{asm:key_upd} by a computer of $\Upsilon$~floating point number operations per second (FLOPS), that is, 
    \begin{equation}
        \tau(N, \lambda) \coloneqq 2^\lambda N \Upsilon^{-1}.
        \label{eq:sdt}
    \end{equation}
    Note that the sample deciphering time for \textit{typical} homomorphic encryption is given as $\tau(1, \lambda)$ because the same key pair is used for encrypting all data.
\end{definition}

By these definitions, the sample identifying complexity and deciphering time quantify the difficulty of estimating $A$ in \eqref{eq:system} using data of sample size $N$ and the required computation time for recovering the data from ciphertexts, respectively.
Now we introduce two constants, \textit{acceptable estimation error} $\gamma_c$ and \textit{defense period} $\tau_c$, to represent an estimation error acceptable by a defender who is the designer of an encrypted control system and a period in which the system is desired to be protected.
Combining with $\gamma$, $\tau$, $\gamma_c$, and $\tau_c$, the security of encrypted control systems can be defined as follows.

\begin{definition}
\label{def:security}
    Let $\gamma_c$ be an acceptable estimation error and let $\tau_c$ be a defense period.
    The encrypted control system \eqref{eq:ecs} is secure if there does not exist a sample size $N$ such that $\gamma(N, F) < \gamma_c$ and $\tau(N, \lambda) \le \tau_c$, where $\gamma$ and $\tau$ are defined in \defref{def:sic} and \defref{def:sdt}, respectively.
    Otherwise, \eqref{eq:ecs} is unsecure.
\end{definition}

A larger $\gamma_c$ and longer $\tau_c$ imply a more secure encrypted control system as long as the system is secure.
In other words, a pair $(\gamma_c, \tau_c)$ represents the security level of a secure encrypted control system.
The constants are later used as design parameters for the optimal security parameter.

\begin{remark}
    Let $\delta > 0$.
    $|A_{ij} - \hat{A}_{ij}| \ge \delta$ holds for all $i,j = 1, \dots, n$ only if $\epsilon(N, F) \ge \delta^2$, where $A$, $\hat{A}$, and $\epsilon$ are defined in \eqref{eq:system}, \eqref{eq:estimate}, and \eqref{eq:error}, respectively.
    This fact suggests that $\gamma_c = \delta^2$ is one of the reasonable choices for an acceptable estimation error.
    Note that $\delta$ should be tailored for a given control system based on its potential risk.
    Additionally, a defense period $\tau_c$ can be chosen as a life span of \eqref{eq:cs}.
\end{remark}

\begin{remark}
    For a noiseless case, i.e., $w_t = 0$, an attacker can exactly identify $A$ of \eqref{eq:system} by decipering $n+1$ samples since $A$ is an $n$-by-$n$ matrix.
    The security definition in such a case can be modified so that \eqref{eq:ecs} is secure if $\tau(n + 1, \lambda) >\tau_c$.
\end{remark}

\section{Encrypted Control System Design}
\label{sec:optimal_design}

This section presents a design method for an optimal controller and security parameter.
To this end, we propose a novel sample identifying complexity of \eqref{eq:system} under the attack in \defref{def:attack}.
We reveal that the optimal controller can be designed as an $H_2$ optimal controller maximizing the sample identifying complexity.
Subsequently, the optimal security parameter is determined using the controllability Gramian of \eqref{eq:system} with the optimal controller.

\subsection{Optimal controller}

The sample identifying complexity of \eqref{eq:system} under the attack in \defref{def:attack} is obtained as follows.

\begin{lemma}\label{lem:sic}
    The function
    \begin{equation}
        \gamma(N, F) \coloneqq n [(N-1) \tr(\Psi)]^{-1}
        \label{eq:sic}
    \end{equation}
    is a sample identifying complexity of \eqref{eq:system} under the attack in \defref{def:attack}, where $\Psi = \Psi(F)$ is a solution to the discrete Lyapunov equation $A \Psi A^\top - \Psi + I = 0$.
\end{lemma}

\begin{proof}
    It follows from \eqref{eq:estimate} and \eqref{eq:error} that $\EV[\epsilon(N, F)] = (1 / n^2) \EV[ \| W_p X_p^+ \|_F^2 ] = (1 / n^2) \EV[ \tr( X_p^+ (X_p^+)^\top W_p^\top W_p ) ]$.
    Let $\bar{X} = X_p^+ (X_p^+)^\top$, and let $\bar{W} = W_p^\top W_p$.
    Then, we obtain
    \begin{align*}
                    \EV\!\left[ \tr\!\left( \bar{X} \bar{W} \right) \!\right]
        &\!=\!      \EV\!\left[ \left( \bar{X}_{11} \bar{W}_{11} \!+\! \cdots \!+\! \bar{X}_{1,T} \bar{W}_{T,1} \right) \right. \\
        &\quad\quad \left. \!+\! \cdots \!+\! \left( \bar{X}_{T,1} \bar{W}_{1,T} \!+\! \cdots \!+\! \bar{X}_{T,T} \bar{W}_{T,T} \right) \right]\!, \\
        &\!=\!      \EV\!\left[ \sum_{j=1}^{t_f-t_s} \sum_{k=1}^{t_f-t_s} \bar{X}_{jk} w_{t_s-1+k}^\top w_{t_s-1+j} \right]\!, \\
        &\!=\!      \EV\!\left[ \sum_{k=1}^{t_f-t_s} \bar{X}_{kk} w_{t_s-1+k}^\top w_{t_s-1+k} \right]\!, \\
        &\!=\!      \EV\!\left[ \tr\!\left( \bar{X} \diag\!\left(\! w_{t_s}^\top w_{t_s}, \dots, w_{t_f-1}^\top w_{t_f-1} \!\right) \!\right) \!\right]\!, \\
        &\!=\!      \tr\!\left( \!\EV\!\left[ \bar{X} \right] \! \EV\!\left[ \diag\!\left(\! w_{t_s}^{\!\top} w_{t_s}, \dots, w_{t_{\!f}-1}^{\!\top} w_{t_{\!f}-1} \!\right) \!\right] \!\right)\!, \\
        &\!=\!      n \sigma^2 \tr\!\left( \EV\!\left[ X_p^{\!\top} \! ( X_p X_p^{\!\top} )^{-1} (X_p^{\!\top} \! ( X_p X_p^{\!\top} )^{-1})^{\!\top} \right] \right)\!, \\
        &\!=\!      n \sigma^2 \tr\!\left( \EV\!\left[ (X_p X_p^\top)^{-1} \right] \right)\!,
    \end{align*}
    where $T \!=\! t_f \!-\! t_s$, $\bar{W}_{ij} \!=\! w_{t_s-1+i}^\top w_{t_s-1+j}$, and the third equality follows from that $w_{t_s-1+k}$ and $w_{t_s-1+j}$ are independent for $j \!\ne\! k$.
    From Jensen's inequality, $\tr\left( (X_p X_p^\top)^{-1} \right) = \sum_{i=1}^n \lambda_i\left( (X_p X_p^\top)^{-1} \right) = \sum_{i=1}^n \lambda_i(X_p X_p^\top)^{-1} = n \sum_{i=1}^n (1/n) \lambda_i(X_p X_p^\top)^{-1} \ge n \left( \sum_{i=1}^n (1/n) \lambda_i(X_p X_p^\top) \right)^{-1} = n^2 \left( \sum_{i=1}^n \lambda_i(X_p X_p^\top) \right)^{-1} = n^2 \tr\left( X_p X_p^\top \right)^{-1}$ and $\EV[X^{-1}] \ge \EV[X]^{-1}$ hold, where $\lambda_i(M)$ denotes the $i$th eigenvalue of $M \in \R^{n \times n}$, and $X$ is a random variable.
    Hence, the trace of the expectation of the inverse matrix is bounded from below by $\tr\left( \EV\left[ (X_p X_p^\top)^{-1} \right] \right) \ge n^2 \EV\left[ \tr\left( X_p X_p^\top \right)^{-1} \right] \ge n^2 \EV\left[ \tr\left( X_p X_p^\top \right) \right]^{-1}$.
    Furthermore, it follows from \eqref{eq:system} that
    \begin{align*}
        &       \EV\left[ \tr\left( X_p X_p^\top \right) \right]
        =       \EV\left[ \tr\left( \sum_{t=t_s}^{t_f-1} x_t x_t^\top \right) \right]\!, \\
        &=      \EV\!\!\left[ \sum_{t=t_s}^{t_f-1} \!\tr \!\left(\!\! A^t x_0 x_0^\top \!(A^t)\!^\top \!\!+\! \sum_{k=0}^{t-1} \! A^{t-1-k} w_k w_k^\top \!(A^{t-1-k})\!^\top \!\!\right)\!\! \right]\!, \\
        &=      \sigma^2 \left[ \sum_{t=t_s}^{t_f-1} \tr\left(A^t (A^t)^\top + \sum_{k=0}^{t-1} A^{t-1-k} (A^{t-1-k})^\top \right) \right]\!, \\
        &=      \sigma^2 \left[ \sum_{t=t_s}^{t_f-1} \tr\left( \sum_{k=0}^t A^k (A^k)^\top \right) \right]\!, \\
        &\le    \sigma^2 \left[ \sum_{t=t_s}^{t_f-1} \tr(\Psi) \right] = \sigma^2 (N-1) \tr(\Psi),
    \end{align*}
    where $A^t (A^t)^\top + \sum_{k=0}^{t-1} A^{t-1-k} (A^{t-1-k})^\top = A^t (A^t)^\top + \sum_{k=0}^{t-1} A^k (A^k)^\top = \sum_{k=0}^t A^k (A^k)^\top \le \sum_{k=0}^\infty A^k (A^k)^\top = \Psi$.
    Consequently, we obtain $\EV[\epsilon(N, F)] \ge (1 / n^2) \cdot n \sigma^2 \cdot n^2 \cdot [\sigma^2 (N-1) \tr(\Psi)]^{-1} = \gamma(N, F)$.
    By \defref{def:sic}, $\gamma(N, F)$ is a sample identifying complexity of \eqref{eq:system} under the attack in \defref{def:attack}.
\end{proof}

The sample identifying complexity \eqref{eq:sic} is computed from the sample size $N$ and the controllability Gramian $\Psi$ of \eqref{eq:system}, which is a function of the feedback gain of \eqref{eq:ctrl}.
By \defref{def:security}, for some feedback gain $F$, the encrypted control system \eqref{eq:ecs} is secure if $\tau(N', \lambda) > \tau_c$ holds for the minimum sample size $N'$ satisfying $\gamma(N', F) < \gamma_c$ because the sample identifying complexity $\gamma(N, F)$ is monotonically decreasing on $N$.
By \defref{def:sdt}, for some security parameter $\lambda$, the sample deciphering time $\tau(N', \lambda)$ increases as $N'$ increases.
Hence, increasing $N'$, the defense period $\tau_c$ can be extended while maintaining the security.
A feedback gain $F$ that maximizes $N'$ can be designed by maximizing $\gamma(N', F)$, i.e., minimizing the trace of $\Psi$.
The following theorem reveals that such a controller is the optimal $H_2$ controller when $\gamma$ is given as \eqref{eq:sic}.

\begin{theorem}\label{thm:opt_ctrl}
    The feedback gain of \eqref{eq:ctrl} maximizing \eqref{eq:sic} is
    \begin{equation}
        F^\ast = Q^\ast (P^\ast)^{-1},
        \label{eq:opt_ctrl}
    \end{equation}
    where $(\eta^\ast, P^\ast, Q^\ast) \in \R \times \R^{n \times n} \times \R^{m \times n}$ is a solution to the problem
    \[
        \min_{(\eta, P, Q)} \eta \ \text{s.t.}\ \tr(P) \,<\, \eta, \ P = P^\top > 0, \  
        \begin{bmatrix}
            P      \!&\! R \!&\! I \\
            R^\top \!&\! P \!&\! O \\
            I      \!&\! O \!&\! I
        \end{bmatrix} > 0,
    \]
    $R = A_p P + B_p Q$, $(A_p, B_p)$ are defined in \eqref{eq:plant}, and $I \in \R^{n \times n}$ and $O \in \R^{n \times n}$ are the identity and zero matrices, respectively.
\end{theorem}

\begin{proof}
    The parameter in \eqref{eq:sic} depending on a feedback gain $F$ is only the Gramian $\Psi = \Psi(F)$.
    Hence, the feedback gain $F^\ast$ maximizing \eqref{eq:sic} satisfies $F^\ast = \argmin{}_F \tr(\Psi)$.
    Now we consider the fictitious system $G: z_{t+1} = A z_t + v_t,\ y_t = z_t$, where $z_t \in \R^n$, $v_t \in \R^n$, and $y_t \in \R^n$.
    Then, $\tr(\Psi) = \|G\|_{H_2}^2$ holds, where $\|\cdot\|_{H_2}$ is the $H_2$ norm, because $\Psi$ is the output controllability Gramian of $G$.
    Therefore, the feedback gain is designed as $F^\ast = \argmin{}_F \|G\|_{H_2} = Q^\ast (P^\ast)^{-1}$, where the second equality follows from~\cite[Proposition~II.1]{steentjes2020}, a discrete-time version of~\cite[Proposition~3.13]{Scherer2000}.
\end{proof}

The controller \eqref{eq:ctrl} with the feedback gain \eqref{eq:opt_ctrl} is the optimal $H_2$ controller for the fictitious system generated by \eqref{eq:system}.
The controller is also optimal for security in the sense of \defref{def:security} when choosing $\gamma$ as \eqref{eq:sic} because the controller maximizes the configurable range of an acceptable estimation error or defense period while maintaining security.
Meanwhile, the optimal controller can reduce the security parameter satisfying the security for some defense period.
This fact is used for designing a security parameter in the next section.

\subsection{Optimal security parameter}

A large security parameter not only improves the security level of encrypted control systems but also generally increases the computation costs of encryption, decryption, and homomorphic evaluation algorithms.
Hence, the minimum security parameter achieving the security is optimal in terms of the implementation costs of encrypted control systems.
Similarly to the optimal controller design, the optimal security parameter can be designed by maximizing the minimum sample size $N'$ satisfying $\gamma(N', F) < \gamma_c$ for some $F$ because, by \defref{def:sdt}, $\lambda$ decreases as $N$ increases for some $\tau(N, \lambda)$.
The maximization is achieved by the optimal controller \eqref{eq:opt_ctrl} as already discussed.
The following lemma shows the minimum sample size with the optimal controller.

\begin{lemma}\label{lem:opt_sample}
    Consider the attack in \defref{def:attack}.
    Given the controller \eqref{eq:ctrl} with the feedback gain \eqref{eq:opt_ctrl}.
    The minimum sample size $N^\ast$ satisfying $\gamma(N^\ast, F^\ast) < \gamma_c$ is
    \begin{equation}
        N^\ast = N(\gamma_c, \Psi^\ast) \coloneqq \left\lfloor n [\gamma_c \tr(\Psi^\ast)]^{-1} \right\rfloor + 2,
        \label{eq:opt_sample}
    \end{equation}
    where $\gamma_c$ is definined in \defref{def:security}, and $\Psi^\ast = \Psi(F^\ast)$ is the Gramian in \lemref{lem:sic} with the feedback gain $F^\ast$.
\end{lemma}

\begin{proof}
    It follows from \eqref{eq:sic} that $\gamma(N, F^\ast) < \gamma_c \!\!\!\iff\!\!\! N > n [\gamma_c \tr(\Psi^\ast)]^{-1} + 1$.
    Therefore, the minimum sample size $N^\ast$ satisfying $\gamma(N^\ast, F^\ast) < \gamma_c$ is given as \eqref{eq:opt_sample}.
\end{proof}

Using the minimum sample size, the optimal security parameter is determined as follows.

\begin{theorem}\label{thm:opt_sec}
    Consider the attack in \defref{def:attack}.
    Given the controller \eqref{eq:ctrl} with the feedback gain \eqref{eq:opt_ctrl}.
    The minimum security parameter $\lambda^\ast$ making the encrypted control system \eqref{eq:ecs} secure, in the sense of \defref{def:security}, is
    \begin{equation}
        \lambda^\ast = \lambda(\tau_c, \Upsilon, N^\ast) \coloneqq \left\lfloor \log_2 \Upsilon \tau_c (N^\ast)^{-1} \right\rfloor + 1,
        \label{eq:opt_sec}
    \end{equation}
    where $\Upsilon$, $\tau_c$, and $N^\ast$ are defined in \defref{def:sdt}, \defref{def:security}, and \lemref{lem:opt_sample}, respectively.
\end{theorem}

\begin{proof}
    It follows from \eqref{eq:sdt} that $\tau(N^\ast, \lambda) > \tau_c \!\!\!\iff\!\!\! \lambda \!>\! \log_2\! \Upsilon \tau_c (N^\ast)^{-1}$.
    Therefore, the minimum security parameter $\lambda^\ast$ satisfying $\tau(N^\ast, \lambda^\ast) > \tau_c$ is given as \eqref{eq:opt_sec}.
\end{proof}

Consequently, the optimal encrypted control system with updatable homomorphic encryption under the attack in \defref{def:attack} can be systematically designed as follows:
1) Set the desired security level $(\gamma_c, \tau_c)$.
2) Suppose an attacker's computer performance $\Upsilon$.   
3) Compute the optimal controller $F^\ast$ of \eqref{eq:opt_ctrl}.
4) Compute the Gramian $\Psi^\ast = \Psi(F^\ast)$ in \lemref{lem:sic}.
5) Compute the sample size $N^\ast = N(\gamma_c, \Psi^\ast)$ of \eqref{eq:opt_sample}.
6) Compute the optimal security parameter $\lambda^\ast = \lambda(\tau_c, \Upsilon, N^\ast)$ of \eqref{eq:opt_sec}.
From \thmref{thm:opt_sec}, the designed encrypted control system is secure, in the sense of \defref{def:security}.
Note that the design procedure can be applied to other attacks by changing the sample identifying complexity \eqref{eq:sic} and controller \eqref{eq:opt_ctrl}.
Moreover, the optimal design of encrypted control systems with typical homomorphic encryption can be achieved by computing only the optimal security parameter $\lambda^\ast_0 = \lambda(\tau_c, \Upsilon, 1)$.

\section{Numerical Example}
\label{sec:example}

Given the parameters of \eqref{eq:plant} as
\[
    A_p =
    \begin{bmatrix}
        0.2 &  0.6 &    0 &   0 \\
        0.5 & -0.5 & -0.1 & 0.2 \\
          0 &    0 &  0.5 &   0 \\
          0 &    0 &    0 & 0.3
    \end{bmatrix},\ 
    B_p =
    \begin{bmatrix}
          0 &   1 \\
          0 &   0 \\
        0.5 & 0.5 \\
          1 &   0
    \end{bmatrix}.
\]
The optimal feedback gain \eqref{eq:opt_ctrl} is given as
\[
    F^\ast =
    \begin{bmatrix}
         0.06 &  0.08 & -0.17 & -0.24 \\
        -0.06 & -0.63 & -0.15 &  0.08
    \end{bmatrix},
\]
where CVXPY~\cite{Diamond16} is used for solving the optimization problem in \thmref{thm:opt_ctrl}.
\figref{fig:error} depicts the estimation error \eqref{eq:error} (gray dots), its expectation (blue solid line), and the sample identifying complexity \eqref{eq:sic} (orange dashed line) for $N = 500, \, \dots, \, 5000$, where the attack of \defref{def:attack} is performed $50$ times for each sample size with $\sigma^2 = 0.01$.
The result shows that \eqref{eq:sic} is an appropriate choice of a sample identifying complexity.
Moreover, with $\gamma_c = 10^{-6}$, $\tau_c = 31536 \times 10^4$~s ($10$ years), and $\Upsilon = 4.42 \times 10^{17}$~FLOPS\footnote{Supercomputer Fugaku. https://www.top500.org/system/179807/}, the minimum sample size \eqref{eq:opt_sample} and security parameter \eqref{eq:opt_sec} are obtained as $N^\ast = 785569$ and $\lambda^\ast = 68$~bit, respectively.
Note that the minimum security parameter making the encrypted control system in \remref{rem:ecs} secure is $\lambda^\ast_0 = \lambda(\tau_c, \Upsilon, 1) = 87$~bit.
When using the updatable homomorphic encryption in~\cite{teranishi2023a} and the ElGamal encryption~\cite{ElGamal85}, the minimum key lengths \eqref{eq:opt_key} achieving $\lambda^\ast$ and $\lambda^\ast_0$ bit security are respectively given as $k^\ast = 589$ and $1031$~bit, where the time complexity of the fastest known algorithm for breaking the encryption schemes is $\Omega(k) = \exp\{(64/9)^{1/3} (\ln 2^k)^{1/3} (\ln\ln 2^k)^{2/3}\}$~\cite{Bernstein93}.

\begin{figure}[t]
    \centering
    \includegraphics[scale=1]{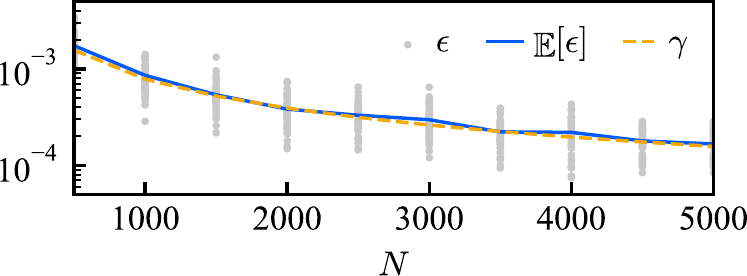}
    \caption{Estimation error and sample identifying complexity.}
    \label{fig:error}
\end{figure}

\section{Conclusions}
\label{sec:conclusions}

This study proposed an optimal controller and security parameter for encrypted control systems under the least squares identification, disclosing the parameters of a closed-loop system.
We revealed that the optimal controller is an $H_2$ optimal controller, and the optimal security parameter was computed for an encrypted control system with the controller and updatable homomorphic encryption.

Updatable homomorphic encryption plays a crucial role in the proposed design method.
In future work, we will construct updatable (leveled) fully homomorphic encryption.
Furthermore, the proposed design method can be applied to other attacks, such as subspace identification.

\bibliographystyle{IEEEtran}
\bibliography{encrypted_control_and_optimization,others}

\end{document}